%% file: ijcai19.tex
\documentclass{article}
\usepackage{ijcai19}

\usepackage{times}
\usepackage{soul}
\usepackage{url}
\usepackage[hidelinks]{hyperref}
\usepackage[utf8]{inputenc}
\usepackage[small]{caption}
\usepackage{graphicx}
\usepackage{amsmath}
\usepackage{amsthm}
\usepackage{amssymb}
\usepackage{amsfonts}
\usepackage{booktabs}
\usepackage{algorithm}
\usepackage[noend]{algpseudocode}
\newtheorem{definition}{\textbf{Definition}}
\newtheorem{lemma}{\textbf{Lemma}}
\newtheorem{theorem}{\textbf{Theorem}}

\title{Scaling Fine-grained Modularity Clustering for Massive Graphs}

\author{
Hiroaki Shiokawa \and Toshiyuki Amagasa \and Hiroyuki Kitagawa
\affiliations
Center for Computational Sciences, University of Tsukuba, Japan
\emails
\{shiokawa, amagasa, kitagawa\}@cs.tsukuba.ac.jp
}

\begin{document}

\maketitle

\begin{abstract}
\input{abstract}
\end{abstract}

\input{introduction}
\input{preliminary}

\input{proposed}

\input{evaluation}

\input{conclusion}
\input{acknowledge}
\input{appendix}

\bibliographystyle{named}
\bibliography{ijcai19}

\end{document}

%% file: abstract.tex
Modularity clustering is an essential tool to understand complicated graphs. 
However, existing methods are not applicable to massive graphs due to two serious weaknesses. 
(1) It is difficult to fully reproduce ground-truth clusters due to the resolution limit problem. 
(2) They are computationally expensive because all nodes and edges must be computed iteratively. 
This paper proposes \textit{gScarf}, which outputs fine-grained clusters within a short running time. 
To overcome the aforementioned weaknesses, gScarf dynamically prunes unnecessary nodes and edges, ensuring that it captures fine-grained clusters. 
Experiments show that gScarf outperforms existing methods in terms of running time while finding clusters with high accuracy.

%% file: introduction.tex
\section{Introduction}
\label{sec:introduction}
Our work is motivated by one question.
How can we efficiently extract fine-grained clusters included in a massive graph?
Modularity clustering \cite{Newman2004} is a fundamental graph analysis tool to understand complicated graphs \cite{TakahashiSK2017,SatoSYK18}.
It detects a set of clusters that maximizes a clustering metric, \textit{modularity} \cite{Newman2004}.
Since greedily maximizing the modularity achieves better clustering results, modularity clustering is employed in various AI-based applications \cite{Louni2018,ShiokawaTK18}.

Although modularity clustering is useful in many applications, it has two serious weaknesses. 
First, it fails to reproduce ground-truth clusters in massive graphs due to the \textit{resolution limit problem} \cite{FortunatoB2007}.
Fortunato \textit{et al.} theoretically proved that the modularity becomes larger until each cluster contains $\sqrt{m}$ edges, where $m$ is the number of edges included in a given graph. 
That is, modularity maximization prefers to find coarse-grained clusters regardless of the ground-truth cluster size.
Second, modularity clustering requires a large computation time to identify clusters since it exhaustively computes all nodes and edges included in a graph.
In the mid-2000s, modularity clustering was applied to social networks with at most thousands of edges \cite{Pella2005}.
By contrast, recent applications must handle massive graphs with millions or even billions of edges \cite{Louni2018} because graphs are becoming larger, and large graphs can be easily found.
As a result, current modularity clustering methods \cite{Blondel2008} need to consume several dozens of hours to obtain clusters from massive graphs. 

\subsection{Existing Approaches and Challenges} 
\label{sec:related}
Many studies have strived to overcome these weaknesses.
One major approach is to avoid the resolution limit effects by modifying the modularity clustering metric.
Locality-aware modularity metrics \cite{MuffRC2005,LiZWZC2008,SankarRS2015} are the most successful ones.
Modularity is not a scale-invariant since it implicitly assumes that each node interacts with all other nodes.
However, it is more reasonable to assume that each node interacts just with its neighbors.
By modifying the modularity so that it refers only to neighbor clusters, the metrics successfully moderate the resolution limit effects for small-sized graphs.
Costa, however, recently pointed out that graph size still affects such metrics \cite{Costa2014}.
That is, the metrics do not fully reproduce ground-truth clusters if graphs are large.

Instead of locality aware metrics, Duan \textit{et al.} proposed a sophisticated method, \textit{correlation-aware modularity clustering} (CorMod) \cite{DuanCDG2014}.
They found that modularity shares the same idea with a correlation measure \textit{leverage} \cite{Piatetsky1991}, 
and they removed biases that cause the resolution limit problem from modularity through a correlation analysis on leverage.
Based on the above analysis, Duan \textit{et al.} defined a different type of modularity named the \textit{likelihood-ratio modularity (LRM)}.
Unlike other modularity metrics, CorMod can avoid producing coarse-grained clusters by maximizing LRM in massive graphs.

Although CorMod can effectively reproduce the ground-truth clusters, it still suffers from the large computational costs to handle massive graphs.
Similar to the original modularity clustering, CorMod iteratively computes all nodes and edges to find clusters that maximize LRM.
This incurs $O(nm\log n)$ time, where $n$ and $m$ are the numbers of nodes and edges, respectively.
Several efficient methods have been proposed for modularity maximization such as the node aggregation and pruning approaches used in Louvain \cite{Blondel2008} and IncMod \cite{ShiokawaFO2013},
but they cannot be directly applied to CorMod since they have no guarantee for the clustering quality in the LRM maximization.
This is because the approaches focus on only the modularity maximization, whereas CorMod employs more complicated function LRM.
Thus, it is a challenging task to improve the computational efficiency of CorMod for fine-grained modularity clustering that approximate the ground-truth clusters.

\subsection{Our Approaches and Contributions} 
We focus on the problem of speeding up CorMod to efficiently extract clusters that well approximate the ground-truth sizes from massive graphs.
We present a novel correlation-aware algorithm, \textit{gScarf}, which is designed to handle billion-edge graphs without sacrificing the clustering quality of CorMod.
The basic idea underlying gScarf is to dynamically remove unnecessary computations for nodes and edges from the clustering procedure.
To determine which computations to exclude, gScarf focuses on the deterministic property of LRM as it is uniquely determined using only the structural properties of clusters such as degrees and cluster sizes.
That is, LRM does not need to be computed repeatedly for clusters with the same structural properties.

Based on the deterministic property, gScarf employs the following techniques to improve efficiency.
(1) gScarf theoretically derives an incremental form of LRM, namely \textit{LRM-gain},
(2) It introduces \textit{LRM-gain caching} in order to skip unnecessary LRM-gain computations based on the deterministic property, and 
(3) it employs \textit{incremental subgraph folding} to further improve the clustering speed.
As a result, our algorithm has the following attractive characteristics:

\begin{itemize}
\item \textbf{Efficiency:} 
  gScarf achieves faster clustering than state-of-the-art methods proposed in the last few years (Section~\ref{sec:efficiency}).
  We theoretically proved that gScarf has a smaller time complexity than the methods (Theorem~\ref{th:efficiency}).
\item \textbf{High accuracy:} 
  gScarf does not sacrifice the clustering quality compared to CorMod, even though gScarf drops computations for nodes and edges (Section~\ref{sec:accuracy}).
  We theoretically confirmed that our approach does not fail to increase LRM during the clustering procedure.
\item \textbf{Easy to deploy:} 
  Our approach does not require any user-specified parameters (Algorithm~\ref{alg:gscarf}).
  Therefore, gScarf provides users with a simple solution for applications using modularity-based algorithms.
\end{itemize}
To the best of our knowledge, gScarf is the first solution that achieves high efficiency and fine-grained clustering results on billion-edge graphs.
gScarf outperforms the state-of-the-art methods by up to three orders of magnitude in terms of clustering time.
For instance, gScarf returns clusters within five minutes for a Twitter graph with 1.46 billion edges.
Although modularity clustering effectively enhances application quality, it has been difficult to apply to massive graphs.
However, gScarf, which is well suited to massive graphs, should improve the quality in a wider range of AI-based applications.

%% file: preliminary.tex
\section{Preliminary}
\label{sec:preliminary}
Here, we formally define the notations and introduce the background.
Let $G = (V, E, W)$ be an undirected graph, where $V$, $E$, and $W$ are sets of nodes, edges, and weights of edges, respectively.
The weight of edge $(i,j) \in E$ is initially set to $W_{i,j}=1$.
Graph clustering divides $G$ into disjoint clusters $C_i = (V_i, E_i)$, in which $V = \bigcup_{i} V_i$ and $V_i \cap V_j = \emptyset$ for any $i \neq j$.
We assume graphs are undirected only to simplify the representations.
Other types of graphs such as directed graphs can be handled with only slight modifications.
For further discussions, we detail how to apply our proposed approaches to the directed graphs in Appendix A.

\subsection{Modularity}
The modularity-based algorithms detect a set of clusters $\mathbb{C}$ from $G$ so that $\mathbb{C}$ maximizes a clustering metric, called \textit{modularity} \cite{Newman2004}.
Modularity measures the differences in graph structures from an expected random graph.
Here, we denote $tp(i) = e_{i}/2m$ and $ep(i) = (a_i/2m)^2$, where $e_{i}$ is a number of edges within $C_i$, $a_{i}$ is a total degree of all nodes in $C_i$, and $m$ is the number of edges in $G$.
Given a set of clusters $\mathbb{C}$, modularity for $\mathbb{C}$ is given by the following function $Q(\mathbb{C})$:
\begin{equation}
  \label{eq:modularity}
  \textstyle
  Q(\mathbb{C}) = \sum_{i} Q(C_i) = \sum_{i} \left\{ tp(i) - ep(i) \right\}.
\end{equation}
\noindent Note that, in Equation~(\ref{eq:modularity}), $ep(i) = (a_i/2m)^2$ indicates the expected fraction of edges, which is obtained by assuming that $G$ is a random graph.
Thus, $Q$ increases when each cluster $C_i$ has a larger fraction of inner-edges $ep(i) = e_i / 2m$ than that of the random graph.

Fortunato and Barthel\'emy theoretically concluded that modularity suffers from the resolution limit problem \cite{FortunatoB2007}.
They proved that $Q(\mathbb{C})$ increases until each cluster includes $\sqrt{2m}$ edges.
That is, if a given graph is too large, the modularity maximization produces many super-clusters.
Thus, modularity maximization fails to reproduce the ground-truth clusters in massive graphs.

\subsection{Likelihood-ratio Modularity}
To overcome the problem, Duan~\textit{et al.} recently proposed \textit{CorMod} \cite{DuanCDG2014}.
CorMod employs the following modularity-based metric \textit{likelihood-ratio modularity} (LRM), which guarantees to avoid the resolution limit:
\begin{equation}
  \label{eq:lrm}
  \textstyle
  LRM(\mathbb{C}) = \sum_{i} LRM(C_i) = \sum_{i} \frac{Pr(tp(i), e_i, 2m)}{Pr(ep(i), e_i, 2m)},
\end{equation}
where $Pr(p, k, n)$ is $\binom{n}{k}p^{k}(1-p)^{n-k}$, which denotes the binomial probability mass function.
$Pr(tp(i), e_i, 2m)$ shows the probability of obtaining $C_i$ from $G$, whereas $Pr(ep(i), e_i, 2m)$ is the expected probability that $C_i$ is taken from a random graph.
As shown in Equation (\ref{eq:lrm}), LRM is designed to balance the size of $Pr(tp(i), e_i, 2m)$ and the ratio of $\frac{Pr(tp(i), e_i, 2m)}{Pr(ep(i), e_i, 2m)}$ for each cluster.
If $C_i$ is small, $Pr(tp(i), e_i, 2m)$ becomes large while the ratio of $\frac{Pr(tp(i), e_i, 2m)}{Pr(ep(i), e_i, 2m)}$ becomes small.
On the other hand, if $C_i$ is large, $Pr(tp(i), e_i, 2m)$ becomes small, the ratio becomes large.
Thus, unlike modularity, LRM successfully avoids to produce super-clusters regardless of the graph sizes.

%% file: proposed.tex
\section{Proposed Method: gScarf}
\label{sec:proposed}
We present gScarf that efficiently detects clusters from massive graphs based on LRM maximization.
First, we overview the ideas underlying gScarf and then give a full description.

\input{ideas}
\input{lrmgain}

\input{caching}
\input{folding}
\input{algorithm}

%% file: ideas.tex
\subsection{Ideas}
\label{subsec:ideas}
Our goal is to efficiently find clusters without sacrificing clustering quality compared with CorMod.
CorMod, which is a state-of-the-art approach, iteratively computes LRM for all pairs of clusters until LRM no longer increases.
By contrast, gScarf skips unnecessary LRM computations by employing the following approaches.
First, we theoretically derive an incremental form of LRM named \textit{LRM-gain}.
LRM-gain is a deterministic criterion to measure the rise of LRM obtained after merging a pair of clusters.
Second, we introduce \textit{LRM-gain caching} to remove duplicate LRM computations that are repeatedly invoked.
Finally, we employ \textit{incremental subgraph folding} to prune unnecessary computations for nodes and edges.
Instead of exhaustive computations for the entire graph, gScarf computes only the essential nodes and edges efficiently to find clusters.

Our ideas have two main advantages.
(1) gScarf finds all clusters with a quite-small running time on real-world graphs.
Our ideas successfully handle the power-law degree distribution, which is a well-known property of real-world graphs \cite{Faloutsos1999}.
This is because gScarf increases its performance if a lot of nodes have similar degrees. 
Hence, the above property lead gScarf to compute efficiently.
(2) gScarf produces almost the same clusters as those of CorMod \cite{DuanCDG2014}.
We theoretically demonstrate that gScarf does not miss chances to improve LRM, although gScarf dynamically prunes nodes and edges from a graph.
Thus, gScarf does not sacrifice clustering quality compared to CorMod.

%% file: lrmgain.tex
\subsection{Incremental LRM Computation}
\label{sec:lrmgain}
We define \textit{LRM-gain}, which measures the rise of the LRM scores obtained after merging two clusters.
\begin{definition}[LRM-gain $\triangle L_{i,j}$]
\label{def:lrmgain}
Let $\triangle L_{i,j}$ be the gain of LRM obtained by merging two clusters, $C_i$ and $C_j$, $\triangle L_{i,j}$ is given as follows:
\begin{equation}
  \label{eq:lrmgain}
  \textstyle
  \triangle L_{i,j} = \triangle P_{i,j} - \triangle Q_{i,j},
\end{equation}
where $\triangle P_{i,j}$ and $\triangle Q_{i,j}$ are the gains of the probability ratio and the modularity, respectively.
Let $C_{(i,j)}$ be the cluster obtained by merging $C_i$ and $C_j$, and their definitions are given as follows:
\begin{eqnarray}
  \textstyle \triangle P_{i,j} &=& P(C_{(i,j)}) - P(C_i) - P(C_j), \\
  \textstyle \triangle Q_{i,j} &=& Q(C_{(i,j)}) - Q(C_i) - Q(C_j),
\end{eqnarray}
where $P(C_i)$ is $tp(i)\ln\frac{tp(i)}{ep(i)}$ if $tp(i)>0$; otherwise, $0$.
\end{definition}
\noindent Note that $\triangle P_{i,j}$ represents a gain of the probability-ratio~\cite{Brin1997}, which is a well-known correlation measure.
Additionally, $\triangle Q_{i,j}$ is the modularity-gain \cite{ClausetNM2004} that measures the increase of the modularity after merging $C_i$ and $C_j$ into the same cluster $C_{(i,j)}$.
From Definition~\ref{def:lrmgain}, we have the following property:
\begin{lemma}
  \label{lemma1}
  Let $C_{(i,k)}$ and $C_{(j,k)}$ be clusters obtained by merging $C_k$ into $C_i$ and $C_j$, respectively.
  We always have $LRM(C_{(i,k)}) \ge LRM(C_{(j,k)})$ iff $\triangle L_{i,k} \ge \triangle L_{j,k}$.
\end{lemma}
\begin{proof}
We first transform Equation~(\ref{eq:lrm}) by applying the Poisson limit theorem~\cite{Papoulis2002} to the binomial probability mass functions, $Pr(tp(i), e_i, 2m)$ and $Pr(ep(i), e_i, 2m)$.
Since $tp(i)$ and $ep(i)$ are significantly small, we can transform $LRM(C_i)$ as follows:
\begin{eqnarray}
  \label{eq:poisson}
  \textstyle
  LRM(C_i)\!\!\!\!\! &=& \!\!\!\!\! \textstyle \frac{Pr(tp(i), e_i, 2m)}{Pr(ep(i), e_i, 2m)} \! = \! \textstyle \frac{(2m\cdot tp(i))^{2m\cdot tp(i)}\cdot\mathrm{e}^{-2m\cdot tp(i)}}{(2m\cdot ep(i))^{2m\cdot tp(i)}\cdot\mathrm{e}^{-2m\cdot ep(i)}} \nonumber \\ 
  &=& \!\!\!\!\! \textstyle \left(\frac{tp(i)}{ep(i)}\right)^{2m\cdot tp(i)}\cdot \mathrm{e}^{-2m\cdot Q(C_i)}. 
\end{eqnarray}
By letting $L(C_i) = \frac{1}{2m}\ln LRM(C_i)$, we have, 
\begin{equation}
  \label{eq:approx_lrm}
  L(C_i) = \textstyle tp(i)\ln\frac{tp(i)}{ep(i)} - \{tp(i)-ep(i)\}.
\end{equation}
Clearly, $\triangle L_{i,k} = L(C_{(i,k)}) - L(C_{i}) - L(C_{k})$.

We then prove $LRM(C_{(i,k)})$ $\ge$ $LRM(C_{(j,k)})$ $\Rightarrow$ $\triangle L_{i,k}$ $\ge$ $\triangle L_{j,k}$.
Since $LRM(C_{(i,k)})$ $\ge$ $LRM(C_{(j,k)})$, we clearly have $LRM(C_{(i,k)}) - LRM(C_{i}) - LRM(C_{k})$ $\ge$ $LRM(C_{(j,k)}) - LRM(C_{j}) - LRM(C_{k})$.
Thus, we hold $\frac{LRM(C_{(i,k)})}{LRM(C_{i})LRM(C_{k})}$ $\ge$ $\frac{LRM(C_{(j,k)})}{LRM(C_{j})LRM(C_{k})}$.
By applying $L(C_i) = \frac{1}{2m}\ln LRM(C_i)$ for the above inequality, $L(C_{(i,k)}) - L(C_{i}) - L(C_{k})$ $\ge$  $L(C_{(j,k)}) - L(C_{j}) - L(C_{k})$.
Hence, $LRM(C_{(i,k)}) \ge LRM(C_{(j,k)})$ $\Rightarrow$ $\triangle L_{i,k} \ge \triangle L_{j,k}$.

We omit the proof of $\triangle L_{i,k}$ $\ge$ $\triangle L_{j,k}$ $\Rightarrow$ $LRM(C_{(i,k)})$ $\ge$ $LRM(C_{(j,k)})$.
However, it can be proved in a similar fashion of the above proof.
\end{proof}
\noindent Lemma~\ref{lemma1} implies that the maximization of LRM and LRM-gain are equivalent.
Consequently,gScarf can find clusters that increase LRM by using LRM-gain in a local maximization manner.
We also identify two additional properties that play essential roles to discuss our LRM-gain caching technique shown in Section~\ref{sec:caching}.
\begin{lemma}
  \label{lemma2}
  $\triangle L_{i,j}$ is uniquely determined by $e_i$, $a_i$, $e_j$, $a_j$, and $e_{i,j}$, where $e_{i,j}$ is the number of edges between $i$ and $j$.
\end{lemma}
\begin{proof}
  Since $tp((i,j))$ $=$ $\frac{e_i + 2e_{i,j} + e_{j}}{2m}$ and $ep((i,j))$ $=$ $\left(\frac{a_i+a_j}{2m}\right )^2$, we can obtain $\triangle P_{i,j}$ from $e_i$, $a_i$, $e_j$, $a_j$, and $e_{i,j}$.
  Furthermore, since $\triangle Q_{i,j} = 2\{\frac{e_{i,j}}{2m} - (\frac{a_i}{2m})(\frac{a_j}{2m})\}$~\cite{Blondel2008}, $\triangle Q_{i,j}$ is clearly determined by $a_i$, $a_j$, and $e_{i,j}$.
  Hence, we hold Lemma~\ref{lemma2} by Definition~\ref{def:lrmgain}.
\end{proof}
\begin{lemma}
  \label{lemma3}
  For each pair of clusters $C_i$ and $C_j$, it requires $O(1)$ time to compute its LRM-gain $\triangle L_{i,j}$.
\end{lemma}
\begin{proof}
  Lemma~\ref{lemma3} is trivial from Definition~\ref{def:lrmgain}.
\end{proof}

%% file: caching.tex
\subsection{LRM-gain Caching}
\label{sec:caching}
We introduce \textit{LRM-gain caching} to reduce the number of computed nodes and edges. 
As shown in Lemma~\ref{lemma2}, LRM-gain $\triangle L_{i,j}$ is uniquely determined by a tuple of structural properties $s_{i,j} = \left < e_i, a_i, e_j, a_j, e_{i,j} \right >$.
That is, if $s_{i,j}$ equals to $s_{i^{\prime},j^{\prime}}$, then $\triangle L_{i,j} = \triangle L_{i^{\prime}, j^{\prime}}$.
Hence, once $\triangle L_{i,j}$ and its corresponding structural property $s_{i,j}$ obtained, $\triangle L_{i,j}$ can be reused to compute other pairs of clusters whose structural properties are equivalent to $s_{i,j}$.

To leverage the above deterministic property, gScarf employs \textit{LRM-gain caching} given as follows:
\begin{definition}[LRM-gain caching]
  \label{def:caching}
  Let $h$ be a hash function that stores $\triangle L_{u,v}$ with its corresponding structural property $s_{u,v}$.
  The LRM-gain caching $h(s_{i,j})$ for a structural property $s_{i,j}=\left < e_i, a_i, e_j, a_j, e_{i,j}\right >$ is defined as follows:
  \begin{eqnarray}
    h(s_{i,j}) = \begin{cases}
      \triangle L_{i^{\prime},j^{\prime}} & (\text{If gScarf finds } s_{i^{\prime},j^{\prime}} \equiv s_{i,j} \text{ in } $h$), \\
      \text{null}       & (Otherwise),\\
    \end{cases}
  \end{eqnarray}
  where $s_{i^{\prime},j^{\prime}} \equiv s_{i,j}$ denotes $s_{i^{\prime},j^{\prime}}$ is equivalent to $s_{i,j}$.
\end{definition}
\noindent gScarf skips the computation of $\triangle L_{i,j}$ if $s_{i^{\prime},j^{\prime}} \equiv s_{i,j}$ is in $h$;
otherwise, gScarf computes $\triangle L_{i,j}$ by Definition~\ref{def:lrmgain}.

To discuss the theoretical aspect of Definition~\ref{def:caching}, we introduce a well-known property called the \textit{power-law degree distribution}~\cite{Faloutsos1999}.
Under this property, the frequency of nodes with $k$ neighbors is $p(k)\propto k^{-\gamma}$, where exponent $\gamma$ is a positive constant that represents the skewness of the degree distribution.

Given the power-law degree distribution, LRM-gain caching has the following property according to Lemma~\ref{lemma3}:
\begin{lemma}
  \label{lemma:caching}
  LRM-gain caching requires $O(d^{2\gamma})$ time during the clustering procedure, where $d$ is the average degree.
\end{lemma}
\begin{proof}
  To simplify, we assume that each cluster is a singleton, \textit{i.e.} $C_i = \{i\}$ for any $i \in V$.
  That is, $s_{i,j}$ $\equiv$ $s_{i^{\prime},j^{\prime}}$ iff $a_i=a_{i^{\prime}}$ and $a_{j}=a_{j^{\prime}}$ because we always have $e_i = e_j = 0$ and $e_{i,j} = 1$ for the singleton clusters.
  Since the frequency of nodes with $k$ neighbors is $p(k)\propto k^{-\gamma}$, the expected number of pairs whose structural properties are equivalent to $s_{i,j}$ is $2m \cdot p(a_{i})p(a_{j})$.
  Thus, to cover the whole of a graph, gScarf needs to compute LRM-gain scores at least $O(\frac{2m}{2m \cdot p(a_{i})p(a_{j})}) \approx O(\frac{1}{p(d)^2})$ times during the clustering.
  From Lemma~\ref{lemma3}, each LRM-gain computation needs $O(1)$ time, LRM-gain caching, therefore, requires $O(\frac{1}{p(d)^2}) = O(d^{2\gamma})$ times.
\end{proof}

%% file: folding.tex
\subsection{Incremental Subgraph Folding}
\label{sec:folding}
We introduce \textit{incremental subgraph folding} that removes unnecessary computations.
Real-world graphs generally have high clustering coefficient~\cite{WattsS1998}; they have a lot of triangles in a graph.
However, a high clustering coefficient entails unnecessary computations in CorMod since the triangles create numerous duplicate edges between clusters.
To avoid such computations, gScarf incrementally folds nodes placed in the same cluster into an equivalent node with weighted edges by extending theoretical aspects of the incremental aggregation presented in \cite{ShiokawaFO2013}.

First, we define the subgraph folding as follows:

\begin{definition}[Subgraph Folding]
  \label{def:folding}
  Given $G=(V,E, W)$, $i\in C_i$ and $j \in C_j$, where $C_i = (V_i, E_i)$ and $C_j = (V_j, E_j)$.
  Let $f$ be a function that maps every nodes in $V\backslash\{V_i \cup V_j\}$ to itself; otherwise, maps it to a new node $x$.
  The subgraph folding of nodes $i$ and $j$ results in a new graph $G^{\prime} = (V^{\prime}, E^{\prime}, W^{\prime})$, 
  where $V^{\prime}=V\backslash\{V_i \cup V_j\}\cup\{x\}$ and $E^{\prime} = \{ (f(u), f(v)) | (u,v)\in E \}$ with updated weight values $W^{\prime}_{f(u),f(v)}$ such that
  \begin{equation*}
    \textstyle
    W^{\prime}_{f(u),f(v)} \!\! = \!\! \begin{cases}
      2W_{u,v}+W_{u,u}+W_{v,v} & \!\!\!\! (f(u) =    x, f(v) =    x) \\
      W_{i,v} + W_{j,v}        & \!\!\!\! (f(u) =    x, f(v) \neq x) \\
      W_{u,i} + W_{u,j}        & \!\!\!\! (f(u) \neq x, f(v) =    x) \\
      W_{u,v}                  & \!\!\!\! (f(u) \neq x, f(v) \neq x) \\
    \end{cases}.
  \end{equation*}
\end{definition}
\noindent Definition~\ref{def:folding} transforms two nodes~$i$ and $j$ into an equivalent single node $x$ with weighted edges.
It replaces multiple edges between $i$ and $j$ into a self-loop edge $(x,x) \in E^{\prime}$ with a weight value $W_{x,x}$, which represents the number of edges between $i$ and $j$.
Similarly, let $k \in \Gamma(x)$, it merges edges $(i, k)$ and $(j,k)$ into weighted single edge $(x, k)$.
Thus, the subgraph folding reduces the number of nodes/edges from $G$.

\begin{lemma}
\label{lemma:equivalence}
Let nodes $i, j \in V$ be in the same cluster (\textit{i.e.} $C_i = C_j$).
If gScarf folds the nodes~$i$ and $j$ into a new node~$w$ by Definition~\ref{def:folding}, then LRM taken from node~$w$ is equivalent to that of the $C_{(i,j)}$ composed of nodes~$i$ and~$j$.
\end{lemma}
\begin{proof}
Recall Equation~(\ref{eq:approx_lrm}), we have $L(C_w) = tp(w) \ln \frac{tp(w)}{ep(w)} - \{tp(w) - ep(w)\}$.
Since gScarf folds the nodes~$i$ and~$j$ into node~$w$, $e_{w} = e_i + e_j + 2e_{i,j}$ and $a_w = a_i + a_j$ by Definition~\ref{def:folding}.
Thus, we have 
\begin{eqnarray}
  \label{eq:equivalent}
  \lefteqn{\textstyle L(C_w) = tp(w)\ln\frac{tp(w)}{ep(w)} -\left\{tp(w) - ep(w)\right\} } \nonumber \\
  &=& \!\!\!\!\!\! \textstyle \frac{e_i+e_j+2e_{i,j}}{2m} \ln \frac{\frac{e_i+e_j+2e_{i,j}}{2m}}{\left(\frac{a_i+a_j}{2m}\right)^2} - \frac{e_i+e_j+2e_{i,j}}{2m} + \left(\frac{a_i+a_j}{2m}\right)^2 \nonumber \nonumber \\
  &=& \!\!\!\!\!\! \textstyle \frac{e_{(i,j)}}{2m} \ln \frac{\frac{e_{i,j}}{2m}}{\left(\frac{a_{(i,j)}}{2m}\right)^2} - \frac{e_{(i,j)}}{2m} + \left(\frac{a_{(i,j)}}{2m}\right)^2 = L(C_{(i,j)}).
\end{eqnarray}
By Lemma~\ref{lemma1}, if $L(C_w) = L(C_{(i,j)})$, then we have $LRM(C_w) = LRM(C_{(i,j)})$.
Therefore, from Equation~(\ref{eq:equivalent}), we hold Lemma~\ref{lemma:equivalence}.
\end{proof}
\noindent From Lemma~\ref{lemma:equivalence}, the subgraph folding does not fail to capture LRM-gain that shows a positive score.
That is, gScarf reduces the number of duplicated edges without sacrificing the clustering quality of CorMod.

Based on Definition~\ref{def:folding}, gScarf performs subgraph folding in an incremental manner during the clustering procedure.
gScarf greedily searches clusters using LRM-gain shown in Definition~\ref{def:lrmgain}.
Once cluster $C_i$ is chosen, gScarf computes LRM-gain $\triangle L_{i,j}$ for all neighbor clusters $C_j \in \Gamma(C_i)$, where $\Gamma(C_i)$ denotes a set of clusters neighboring $C_i$.
After that, gScarf folds a pair of clusters that yields the largest positive LRM-gain based on Definition~\ref{def:folding}.

Our incremental subgraph folding has the following theoretical property in terms of the computational costs.
\begin{lemma}
  \label{lemma:folding}
  A subgraph folding entails $O(d)$ times, where $d$ is the average degree.
\end{lemma}
\begin{proof}
By Definition~\ref{def:folding}, the subgraph folding requires to update the weights of each edge by traversing all neighboring clusters.
Thus, a subgraph folding for $C_i$ and $C_j$ entails time-complexity $O(\min\{|\Gamma(C_i)|,|\Gamma(C_j)|\}) = O(d)$.
\end{proof}

%% file: algorithm.tex
\begin{algorithm}[t]
  \caption{Proposed method: gScarf}
  \label{alg:gscarf}
  \footnotesize
  \begin{algorithmic}[1]
    \Require A graph $G=(V, E, W)$;
    \Ensure A set of clusters $\mathbb{C}$;
    %
    \State $\mathbb{T} = \emptyset$;
    \For{\textbf{each} $i \in V$}
      \State $C_i = \{i\}$, and $\mathbb{T} = \mathbb{T}\cup\{C_i\}$;
    \EndFor
    %
    \While{$\mathbb{T}\neq\emptyset$}
      \State Get $C_i$ from $\mathbb{T}$;
      \State $C_{best} = C_i$;
      \For{$C_j \in \Gamma(C_i)$}
        \State $s_{i,j} = \left < e_i, a_i, e_j, a_j, e_{i,j} \right >$;
        \State \textbf{if} $h(s_{i,j}) = null$ \textbf{then} $h(s_{i,j}) \leftarrow \triangle L_{i,j}$;
        \State \textbf{if} $h(s_{i,j})>h(s_{i,best})$ \textbf{then} $C_{best} = C_{j}$;
      \EndFor
      \State $\!\!$ \textbf{if} $h(s_{i,best})>0$ \textbf{then}
      \State \hspace{2mm} $C^{\prime} \leftarrow fold(C_i, C_{best})$ by Definition~\ref{def:folding};
      \State \hspace{2mm} $\mathbb{T} = \mathbb{T}\backslash\{C_{i}, C_{best}\}\cup\{C^{\prime}\}$;
      \State $\!\!$ \textbf{else} $\mathbb{T} = \mathbb{T}\backslash\{C_{i}\}$;
    \EndWhile
    \State $\!\!$ \Return $\mathbb{C}$; 
  \end{algorithmic}
\end{algorithm}

\subsection{Algorithm}
\label{sec:algorithm}
Algorithm~\ref{alg:gscarf} gives a full description of gScarf.
First, gScarf initializes each node as a singleton cluster, \textit{i.e.} $C_{i} = \{i\}$, and it stores all clusters into a target node set $\mathbb{T}$ (lines~1-3).
Then, gScarf starts the clustering phase to find a set of clusters $\mathbb{C}$ that maximizes LRM in a local maximization manner (lines~4-15).
Once gScarf selects a cluster $C_i$ from $\mathbb{T}$ (line~5), it explores neighboring cluster $C_{best}$ that yields the largest positive score of LRM-gain (lines~6-10).
To improve the clustering efficiency, gScarf uses the LRM-gain caching $h(s_{i,j})$ in Definition~\ref{def:caching};
gScarf skips unnecessary computations based on the deterministic property shown in Lemma~\ref{lemma2} (lines~8-10).
If gScarf finds the structural property $s_{i^{\prime},j^{\prime}}$ $\equiv$ $s_{i,j}$ in $h$, then it reuses $h(s_{i^{\prime},j^{\prime}})=\triangle L_{i^{\prime},j^{\prime}}$ 
instead of $\triangle L_{i,j}$; 
otherwise, it computes $\triangle L_{i,j}$ by Definition~\ref{def:lrmgain} (line~9).
After finding $C_{best}$, gScarf performs the incremental subgraph folding for $C_i$ and $C_{best}$ (lines~11-14).
If $h(s_{i,best})=\triangle L_{i,best} > 0$, gScarf contracts $C_i$ and $C_{best}$ into a single cluster $C^{\prime}$ by Definition~\ref{def:folding} (line~19).
gScarf terminates when $\mathbb{T}=\emptyset$ (line~4).
Finally, it outputs a set of clusters $\mathbb{C}$ in (line~15).

The computational cost of gScarf is analyzed as follows:
\begin{theorem}
\label{th:efficiency}
gScarf incurs $O(m + d^{2\gamma})$ time to obtain a clustering result from a graph $G$,
where $m$ is the number of edges in $G$, d is the average degree, and $\gamma$ is a positive constant that controls the skewness of the degree distribution.
\end{theorem}
\begin{proof}
  Algorithm~\ref{alg:gscarf} requires $O(|\mathbb{T}|) \approx O(n)$ iterations for the clustering phase.
  In each iteration, gScarf invokes at most one subgraph folding (lines~18-21) that entails $O(d)$ costs from Lemma~\ref{lemma:folding}.
  Thus, it has $O(nd)=O(m)$ times to perform the clustering phase.
  Additionally, gScarf employs the LRM-gain caching during the iterations.
  As we proved in Lemma~\ref{lemma:caching}, it incurs $O(d^{2\gamma})$ times.
  Therefore, gScarf requires $O(m + d^{2\gamma})$ times to obtain a clustering result.
\end{proof}

\noindent 
The computational costs of gScarf are dramatically smaller than those of CorMod, which requires $O(nm \log n)$ time.
In practice, real-world graphs show $d^{2\gamma}\ll m$ since they generally have small values of $d$ and $\gamma$ due to the power-law degree distribution~\cite{Faloutsos1999}.
Specifically, as shown in Table~\ref{tab:datasets}, the real-world datasets examined in the next section have $d < 39$ and $\gamma<2.3$.
Thus, gScarf has a nearly linear time complexity against the graph size (\textit{i.e.,} $O(m + d^{2\gamma}) \approx O(m)$) on the real-world graphs with the power-law degree distribution.
Furthermore, our incremental subgraph folding effectively handles the clustering coefficient~\cite{WattsS1998,ShiokawaFO2015} to further reduce the computational costs.
Thus, gScarf has an even smaller cost than the one shown in Theorem~\ref{th:efficiency}.

%% file: evaluation.tex
\input{dataset_table}
\section{Experimental Evaluation}
\label{sec:evaluation}
In this section, we experimentally evaluate the efficiency and the clustering accuracy of gScarf.

\begin{figure*}[t]
  \begin{center}
    \includegraphics[height=38mm]{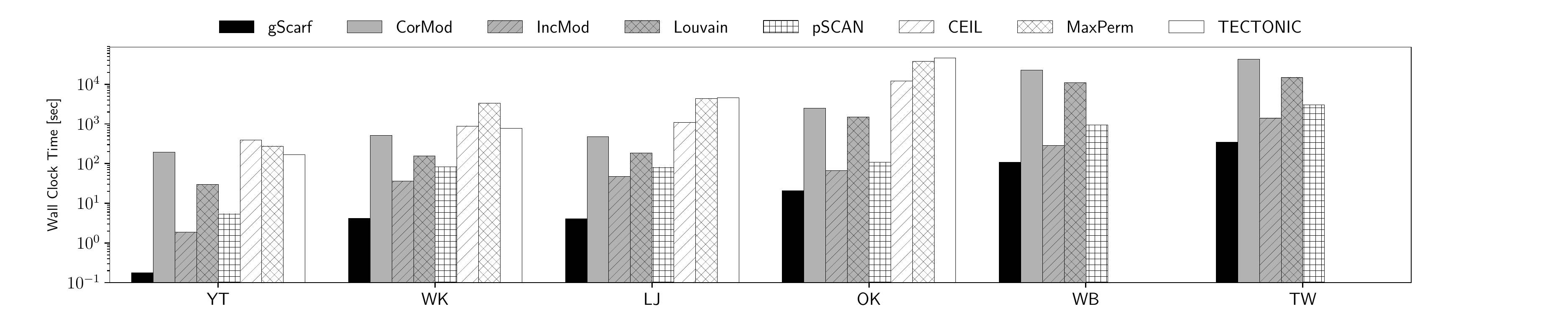}
    \vspace{-8mm}
    \caption{Runtimes on real-world graphs.}
    \label{fig:runtime_real}
  \end{center}
\end{figure*}

\begin{figure*}[t]
  \vspace{-6mm}
  \begin{minipage}{0.49\hsize}
  \begin{center}
    \hspace{2mm}\includegraphics[width=80mm]{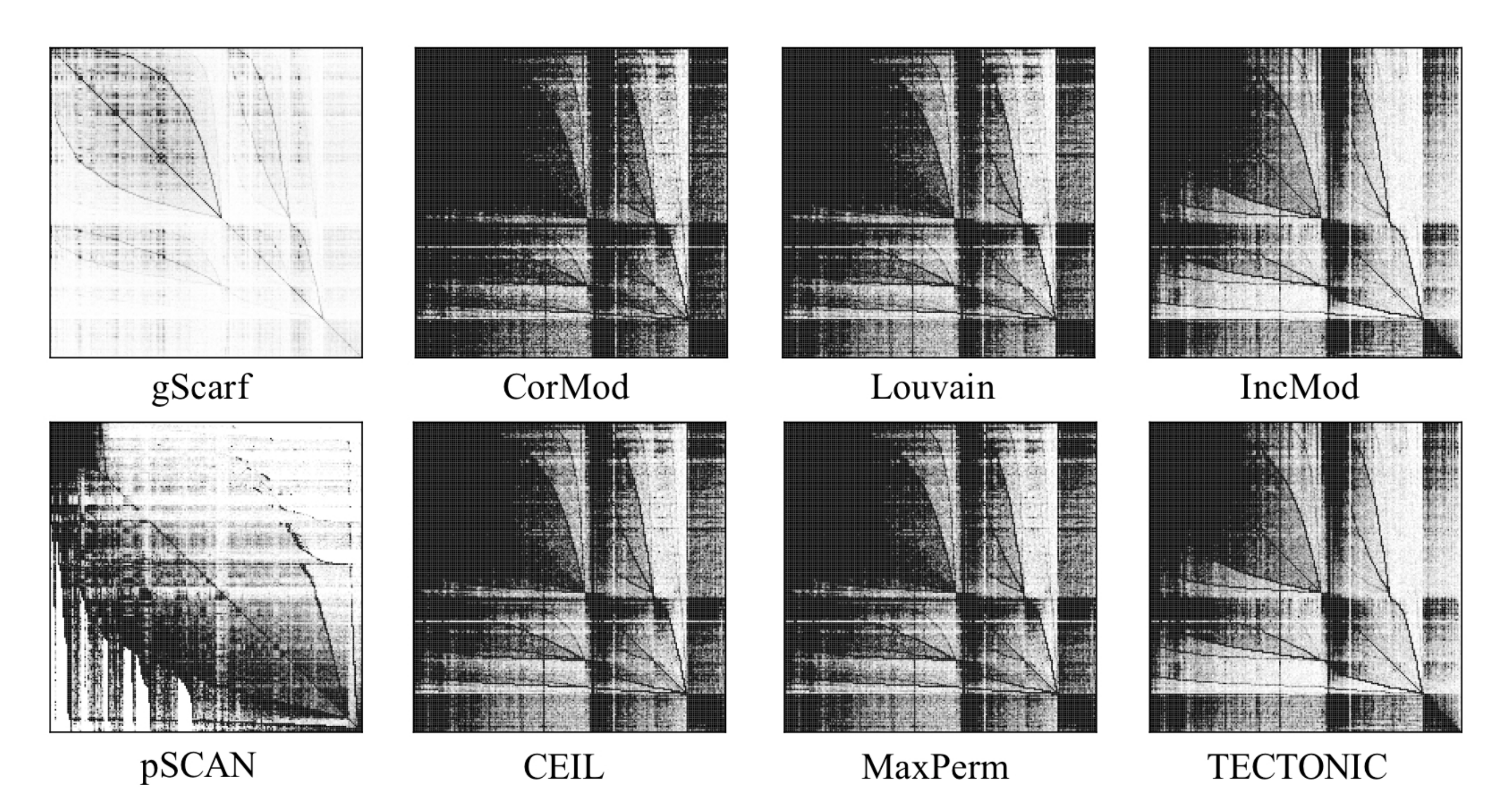}
    \vspace{-4mm}
    \caption{Computed edges on YT in each approach.}
    \label{fig:hist2d}
  \end{center}
  \end{minipage}
  \begin{minipage}{0.49\hsize}
    \begin{center}
      \includegraphics[width=84mm]{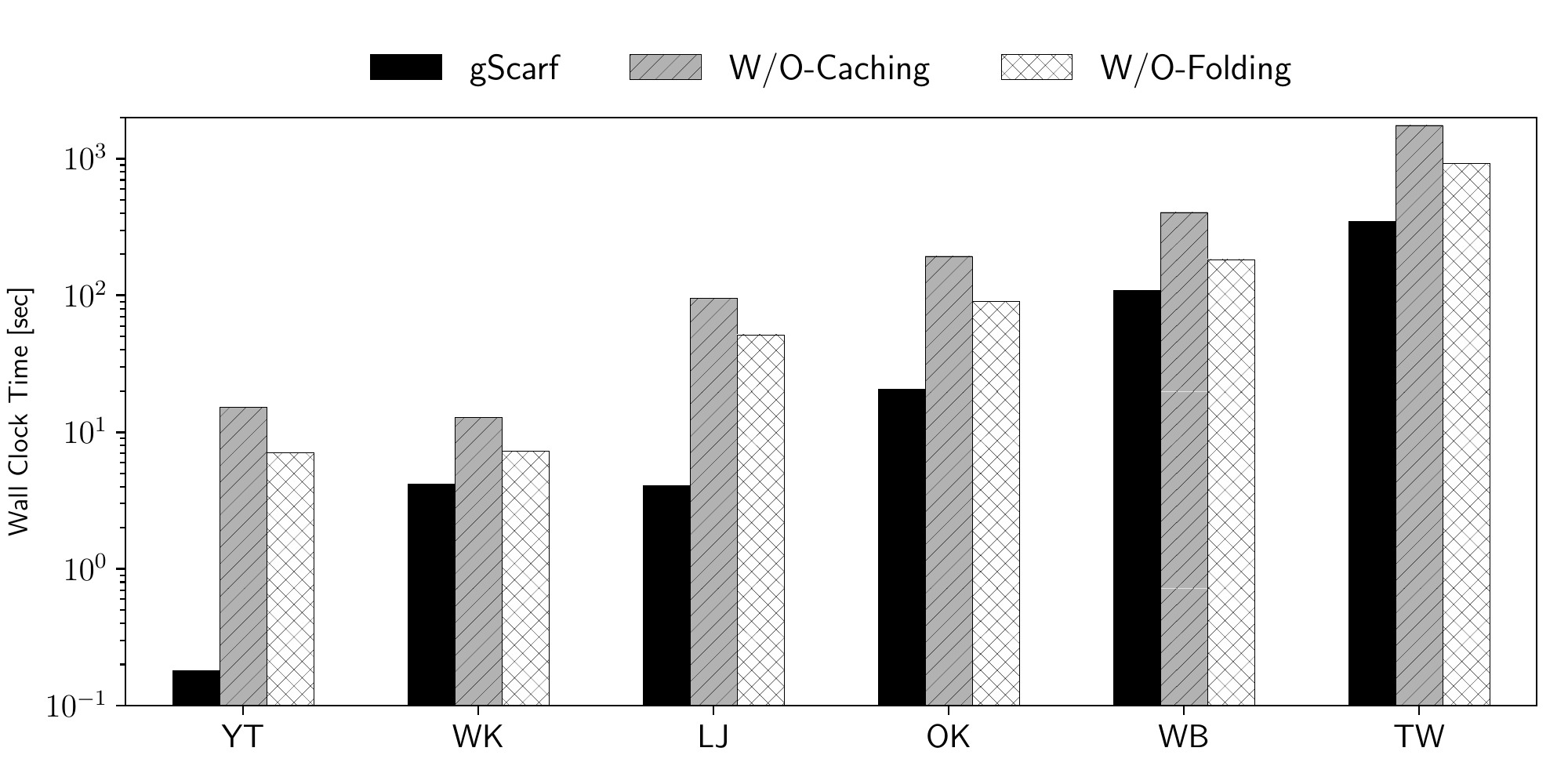}
      \vspace{-4mm}
      \caption{Effectiveness of key approaches.}
      \label{fig:effectiveness}
    \end{center}
  \end{minipage}
  \vspace{-4mm}
\end{figure*}

\noindent \textbf{Experimental Setup:}
We compared gScarf with the following state-of-the-art graph clustering baseline algorithms:
\begin{itemize}
\item \textbf{CorMod}: 
  The original correlation-aware modularity clustering \cite{DuanCDG2014}.
  CorMod greedily maximizes the same criteria (LRM) as gScarf.
\item \textbf{Louvain}: 
  The most popular modularity clustering for large graphs \cite{Blondel2008}.
  This method greedily maximizes the modularity in Equation~(\ref{eq:modularity}).
\item \textbf{IncMod}: 
  The fastest modularity maximization method proposed by \cite{ShiokawaFO2013}.
  It improves the efficiency of Louvain via incremental aggregation and pruning techniques.
\item \textbf{pSCAN}: 
  The density-based graph clustering recently proposed by \cite{ChangLQZY17}.
  It extracts clusters by measuring a density of node connections based on thresholds $\epsilon$ and $\mu$.
  We set $\epsilon = 0.6$ and $\mu = 5$ as the same settings as used in \cite{ChangLQZY17}.
\item \textbf{CEIL}: 
  A scalable resolution limit-free method based on a new clustering metric that quantifies internal and external densities of the cluster \cite{SankarRS2015}.
  It detects clusters by greedily maximizing the metric.
\item \textbf{MaxPerm}: 
  The resolution limit free method based on the localization approach \cite{Tanmoy14}.
  This method quantifies permanence of a node within a cluster, and it greedily explores clusters so that the permanence increases.
\item \textbf{TECTONIC}: 
  The motif-aware clustering method proposed by \cite{Tsourakakis2017}.
  This method extracts connected components as clusters by removing edges with a smaller number of triangles than threshold $\theta$.
  We used the recommended value $\theta = 0.06$.
\end{itemize}
All experiments were conducted on a Linux server with a CPU (Intel(R) Xeon(R) E5-2690 2.60 GHz) and 128 GB RAM.
All algorithms were implemented in C/C++ as a single-threaded program, which uses a single CPU core with the entire graph held in the main memory.

\noindent \textbf{Datasets:} 
We used six real-world graphs published by SNAP \cite{SNAP} and LAW \cite{LAW}.
Table~\ref{tab:datasets} shows their statistics.
The symbols, $d$ and $\gamma$ denote the average degree and the skewness of the power-law degree distribution, respectively.
As shown in Table~\ref{tab:datasets}, only YT, LJ, and OK have their ground-truth clusters.
In our experiments, we also used synthetic graphs with their ground-truth clusters generated by LFR-benchmark \cite{LancichinettiFK2009}, which is the \textit{de facto standard} model for generating graphs.
The setting are described in detail in Section~\ref{sec:accuracy}.
\begin{figure*}[t]
   \begin{minipage}{0.32\hsize}
      \begin{center}
        \includegraphics[width=62mm]{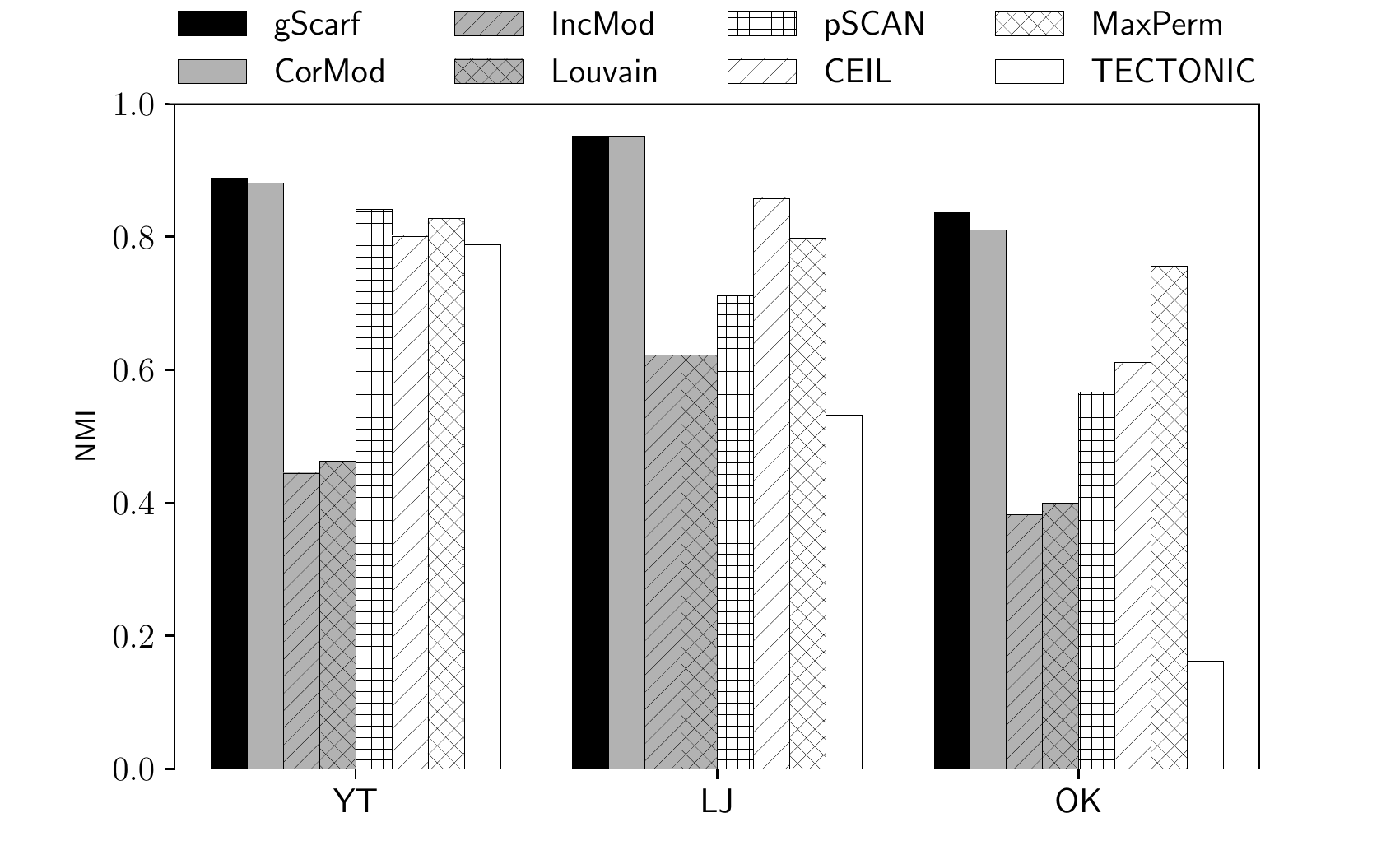}
        \vspace{-1mm}
        \text{\footnotesize (a) Real-world graphs}
      \end{center}
   \end{minipage}
   \begin{minipage}{0.32\hsize}
      \begin{center}
        \includegraphics[width=62mm]{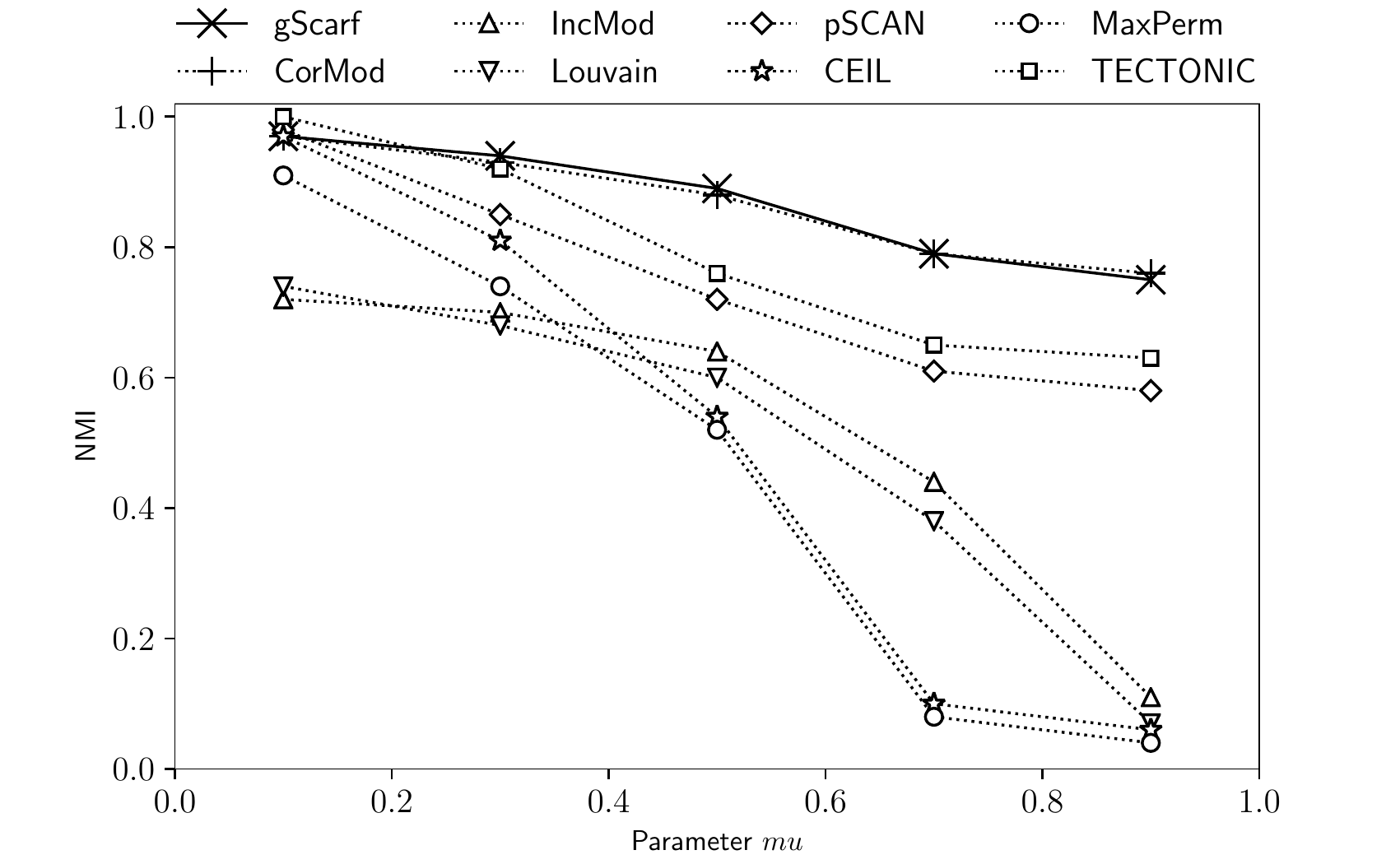}
        \vspace{-1mm}
        \text{\footnotesize (b) LFR benchmark graphs varying $\textit{mu}$.}
      \end{center}
    \end{minipage}
    \begin{minipage}{0.32\hsize}
      \begin{center}
        \includegraphics[width=62mm]{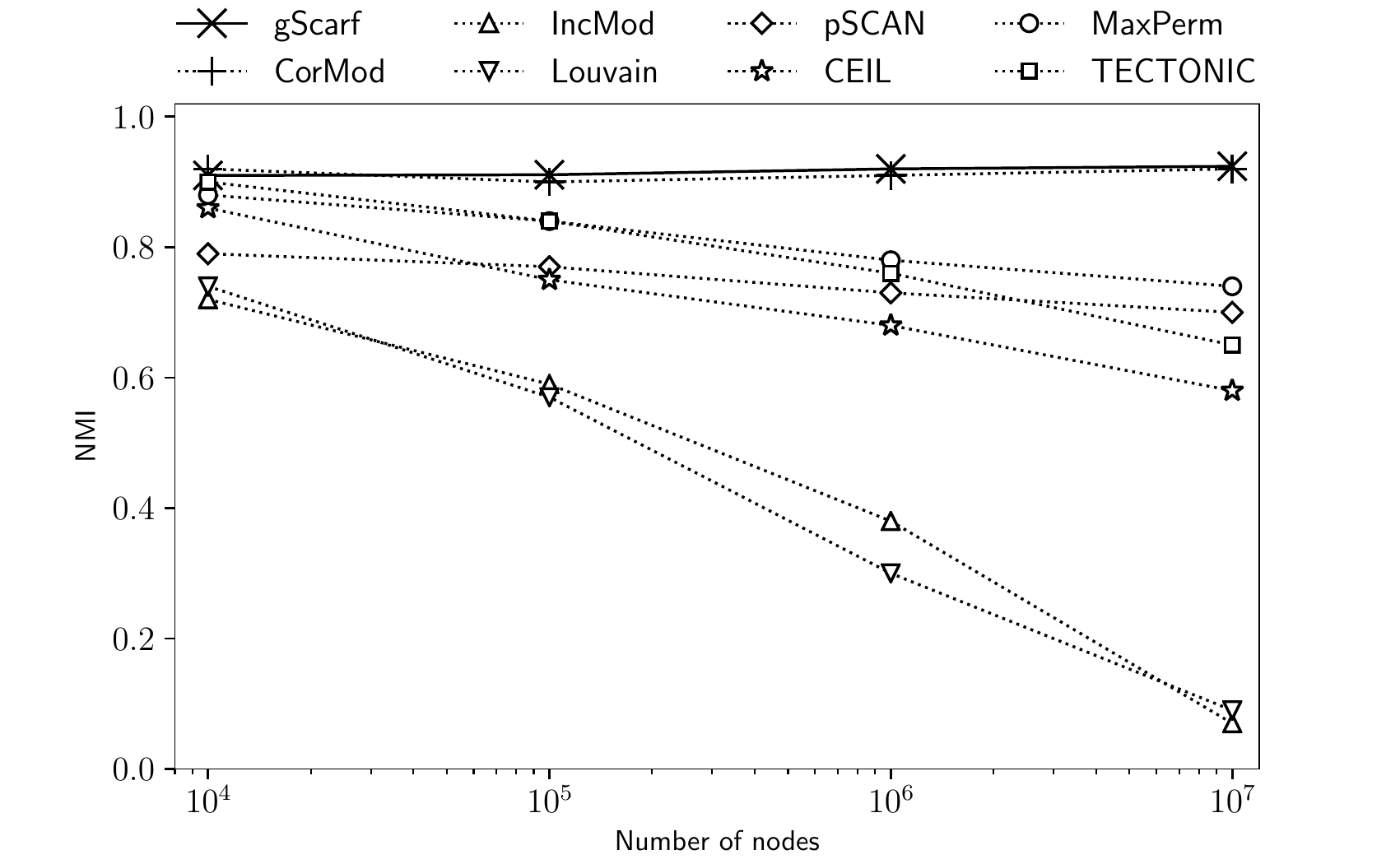}
        \vspace{-1mm}
        \text{\footnotesize (c) LFR benchmark graphs varying $n$.}
      \end{center}
    \end{minipage}
    \vspace{-1mm}
    \caption{NMI scores on real-world and synthetic graphs.}
    \label{fig:accuracy}
    \vspace{-2mm}
\end{figure*}

%
\input{cluster_size}
\subsection{Efficiency}
\label{sec:efficiency}
We evaluated the clustering times of each algorithm in the real-world graphs (Figure~\ref{fig:runtime_real}).
Note that we omitted the results of CEIL, MaxPerm, and TECTONIC for large graphs since they did not finish within 24 hours.
Overall, gScarf outperforms the other algorithms in terms of running time.
On average, gScarf is 273.7 times faster than the other methods.
Specifically, gScarf is up to 1,100 times faster than CorMod, although they optimize the same clustering criterion.
This is because gScarf exploits LRM-gain caching and does not compute all nodes and edges.
In addition, it further reduces redundant computations incurred by triangles included in a graph by incrementally folding subgraphs.
As a result, as we proved in Theorem~\ref{th:efficiency}, gScarf shows a nearly linear running time against the graph size.
That is why our approach is superior to the other methods in terms of running time.

To verify how gScarf effectively reduces the number of computations, Figure~\ref{fig:hist2d} plots the computed edges in each algorithm for YT where a black (white) dot indicates that the (i,j)-th element is (is not) computed.
gScarf decreases the computational costs by 83.9–98.3\% compared to the other methods.
The results indicate that the LRM-gain caching and the incremental subgraph folding successfully reduce unnecessary computations.
This is due to the power-law degree distribution, which is a natural property of real-world graphs including YT.
If a graph has a strong skewness in its degree distribution, there are numerous equivalent structural properties since the graph should have similar degrees.
Hence, gScarf can remove a large amount of computations.

We then evaluated the effectiveness of our key approaches.
We compared the runtimes of gScarf with its variants that exclude either LRM-gain caching or incremental subgraph folding.
Figure~\ref{fig:effectiveness} shows the runtimes of each algorithm in real-world graphs, where \textit{W/O-Caching} and \textit{W/O-Folding} indicate gScarf without LRM-gain caching and incremental subgraph folding, respectively.
gScarf is 21.6 and 10.4 times faster than W/O-Caching and W/O-Folding on average, respectively.
These results indicate that LRM-gain caching contributes more to the efficiency improvements even though it requires $O(d^{2\gamma})$ times.
As discussed in Section~\ref{sec:algorithm}, real-world graphs generally have small values of $d$ and $\gamma$ due to the power-law degree distribution.
For instance, as shown in Table~\ref{tab:datasets}, the datasets that we examined have small $d$ and $\gamma$, \textit{i.e.,} $d \le 38.1$ and $\gamma <2.3$ at most.
Consequently, the time complexity $O(d^{2\gamma})$ is much smaller than the cost for incremental subgraph folding in real-world graphs.
Hence, gScarf can efficiently reduce its running time.

\subsection{Accuracy of Clustering Results}
\label{sec:accuracy}
One major advantage of gScarf is that it does not sacrifice clustering quality, although it dynamically skips unnecessary computations.
To verify the accuracy of gScarf, we evaluated the clustering quality against ground-truth clusters on both real-world and synthetic graphs.
We used \textit{normalized mutual information (NMI)}~\cite{Cilibrasi2005} to measure the accuracy compared with the ground-truth clusters.

\paragraph{Real-world graphs}
Figure~\ref{fig:accuracy}~(a) shows NMI scores of each algorithm in YT, LJ, and OK.
We used the top-5000-community datasets published by \cite{SNAP} as the ground-truth clusters for all graphs.
In the top-5000-community datasets, several nodes belong to multiple ground-truth clusters. 
We thus assigned such nodes into the cluster where most of its neighbor nodes are located.

As shown in Figure~\ref{fig:accuracy}, gScarf shows higher NMI scores than the other methods except for CorMod.
In addition, gScarf has almost the same or slightly higher NMI score than the original method CorMod.
This is because both employ LRM maximization to avoid the resolution limit problem. 
Furthermore, gScarf is theoretically designed not to miss any chance to increase LRM by Lemma~\ref{lemma1} and Lemma~\ref{lemma:equivalence}.
That is, the greedy LRM-gain maximization in gScarf is equivalent to that in CorMod.
Hence, gScarf achieves better NMI scores even though it dynamically prunes unnecessary nodes/edges by LRM-gain caching and incremental subgraph folding.

To analyze the cluster resolution, we compared the average sizes of clusters extracted by each algorithm.
As shown in Table~\ref{tab:size}, gScarf and CorMod effectively reproduce the cluster sizes of the ground-truth in large real-world graphs.
This is because LRM maximization is designed to avoid the resolution limit problem \cite{DuanCDG2014}, and LRM effectively extract cluster sizes that approximate the ground-truth clusters.
On the other hand, other algorithms extract coarse-grained or significantly smaller clusters.
Especially in the modularity maximization methods (\textit{i.e.} Louvain and IncMod), they output larger clusters than the ground-truth since they grow clusters until their sizes reach $\sqrt{m}$.

\paragraph{Synthetic graphs}
To verify the effects of graph structures and graph sizes, we evaluated the NMI scores on LFR-benchmark graphs.
In Figure~\ref{fig:accuracy}~(b), we generated graphs composed of $10^{6}$ nodes with average degree 20 by varying the mixing parameter $\textit{mu}$ from 0.1 to 0.9.
$\textit{mu}$ controls the fraction of neighbor nodes included in other clusters;
as $\textit{mu}$ increases, it becomes more difficult to reveal the intrinsic clusters.
In contrast, in Figure~\ref{fig:accuracy}~(c), we used graphs generated by varying the number of nodes from $10^{4}$ to $10^{7}$, where the average degree and $\textit{mu}$ are set to 20 and 0.5, respectively.

Figure~\ref{fig:accuracy}~(b) shows that gScarf and CorMod achieve high accuracy even if the parameter $\textit{mu}$ increases, whereas most of the other methods degrade NMI scores for large $\textit{mu}$.
In addition, as shown in Figure~\ref{fig:accuracy}~(c), gScarf also shows higher NMI scores than the others for large graphs since LRM effectively avoids to output super-clusters regardless of inputted graph sizes.
Hence, gScarf efficiently detects clusters that effectively approximate the ground-truth clusters in massive graphs.
From these results, we confirmed that gScarf does not sacrifice the clustering quality of CorMod.

%% file: dataset_table.tex
%
\begin{table}[t]
  \begin{center}
    \vspace{-2mm}
    \caption{Statistics of real-world datasets.}
    \vspace{-4mm}
    \label{tab:datasets}
          {\tabcolsep = 1.4mm
            \scriptsize
            \begin{tabular}{|crrrrcl|} \hline 
              \multicolumn{1}{|c}{Name} & \multicolumn{1}{c}{$n$} & \multicolumn{1}{c}{$m$} & \multicolumn{1}{c}{$d$} & \multicolumn{1}{c}{$\gamma$} & \multicolumn{1}{c}{Ground-truth} & \multicolumn{1}{c|}{Source} \\\hline
              YT &                  1.13 M &                  2.98 M &                    2.63 &                         1.93 & \checkmark                       & {\scriptsize com-Youtube (SNAP)} \\ 
              WK &                  1.01 M &                  25.6 M &                    25.1 &                         2.02 & N/A                              & {\scriptsize itwiki-2013 (LAW)} \\ 
              LJ &                  3.99 M &                  34.6 M &                    8.67 &                         2.29 & \checkmark                       & {\scriptsize com-LiveJournal (SNAP)} \\ 
              OK &                  3.07 M &                   117 M &                    38.1 &                         1.89 & \checkmark                       & {\scriptsize com-Orkut (SNAP)} \\ 
              WB &                   118 M &                  1.01 B &                    8.63 &                         2.14 & N/A                              & {\scriptsize webbase-2001 (LAW)} \\ 
              TW &                  41.6 M &                  1.46 B &                    35.2 &                         2.27 & N/A                              & {\scriptsize twitter-2010 (LAW)} \\ \hline
            \end{tabular}
          }
  \end{center}
  \vspace{-4mm}
\end{table}

%% file: cluster_size.tex
%
\begin{table*}[t]
  \begin{center}
    \caption{Average cluster sizes.}
    \label{tab:size}
    \vspace{-3mm}
           {
             \scriptsize
             \begin{tabular}{|c|c|cccccccc|} \hline 
                  & Ground-truth & gScarf & CorMod & Louvain & IncMod  & pSCAN & CEIL & MaxPerm & TECTONIC \\ \hline
               YT & 13.5         &  13.3   & 13.3   & 66.1    & 50.4    & 24.3  & 5.6   & 11.4   & 8.2      \\
               LJ & 40.6         &  44.2   & 45.1   & 111.4   & 104.5   & 81.9  & 11.3  & 33.3   & 10.7     \\
               OK & 215.7        &  194.9  & 194.1  & 16551.9 & 15676.7 & 403.7 & 35.7  & 121.9  & 9.61     \\ \hline
             \end{tabular}
           }
  \end{center}
  \vspace{-5mm}
\end{table*}

%% file: conclusion.tex
\section{Conclusion}
\label{sec:conclusion}
We proposed gScarf, an efficient algorithm that produces fine-grained modularity-based clusters from massive graphs.
gScarf avoids unnecessary computations by LRM-gain caching and incremental subgraph folding.
Experiments showed that gScarf offers an improved efficiency for massive graphs without sacrificing the clustering quality compared to existing approaches.
By providing our efficient approaches that suit for massive graphs, gScarf will enhance the effectiveness of future applications based on the graph analysis.

%% file: acknowledge.tex
\section*{Acknowledgements}
This work was partially supported by JST ACT-I.
We thank to Prof. Ken-ichi Kawarabayashi (National Institute of Informatics, Japan) for his helps and useful discussions.

%% file: appendix.tex
\section*{A How to Handle Directed Graphs}
We here detail gScarf modifications for directed graphs.
As discussed in \cite{Nicosia2009}, we can handle the directed graphs in the modularity maximization 
by replacing $a_{i}$ to $a_{i}^{in} a_{i}^{out}$ in $ep(i)$, where $a_{i}^{in}$ and $a_{i}^{out}$ are in-degree and out-degree of node $i$, respectively.
Thus, our LRM-gain caching can handle the directed graphs by replacing $a_{i}$ and $a_{j}$ of $s_{i,j}$ in Definition~\ref{def:caching} to $a_{i}^{in} a_{i}^{out}$ and $a_{j}^{in} a_{j}^{out}$, respectively.
Similarly, we need to modify Definition~\ref{def:folding} so that the folding technique takes account of directions of edges.
That is, our incremental subgraph folding can handle the directed graphs by aggregating in-degree edges and out-degree separately.